%% file: main_arxiv.tex
\newif\ifready\readytrue
\newif\ifanon\anonfalse

\documentclass[11pt,letterpaper]{article}

\usepackage[utf8]{inputenc}
\input{headers}

\input{macros}
\title{Continual Release of Densest Subgraphs:\\Privacy Amplification \& Sublinear Space via Subsampling}
\ifanon
\author{Anonymous Author(s)}
\else
\author{%
     \begin{tabular}{c}
        Felix Zhou\\
        Yale University\\
        \texttt{felix.zhou@yale.edu}
    \end{tabular}
}
\fi
\date{}

\begin{document}

\maketitle
\sloppy

\ifanon\else
\pagenumbering{gobble}
\fi

\begin{abstract}
    We study the sublinear space continual release model for edge-differentially private (DP) graph algorithms, with a focus on the densest subgraph problem (DSG) in the insertion-only setting. 
    Our main result is the first continual release DSG algorithm 
    that matches the additive error of the best static DP algorithms 
    and the space complexity of the best non-private streaming algorithms, 
    up to constants. 
    The key idea is a refined use of subsampling that simultaneously achieves privacy amplification and sparsification, 
    a connection not previously formalized in graph DP. 
    Via a simple black-box reduction to the static setting,
    we obtain both pure and approximate-DP algorithms with $O(\log n)$ additive error and $O(n\log n)$ space, 
    improving both accuracy and space complexity over the previous state of the art.
    Along the way,
    we introduce \emph{graph densification} in the graph DP setting,
    adding edges to trigger earlier subsampling,
    which removes the extra logarithmic factors in error and space incurred by prior work \cite{epasto2024power}. 
    We believe this simple idea maybe of independent interest.
\end{abstract}

\ifanon
\else
\newpage
\setcounter{tocdepth}{2}
\tableofcontents
\newpage
\cleardoublepage
\pagenumbering{arabic}
\fi

\input{body}

\ifanon
\else
\section*{Acknowledgments}
The author thanks Kunal Talwar for the equivalent interpretation of the subsampling scheme, 
Quanquan C. Liu and Marco Pirazzini for feedback on an initial version of this manuscript,
as well as Jia Shi for several insightful discussions.
Felix Zhou acknowledges the support of the Natural Sciences and Engineering Research Council of Canada (NSERC).
\fi

\begingroup
\sloppy
\printbibliography
\endgroup

\clearpage
\addtocontents{toc}{\protect\setcounter{tocdepth}{1}}
\appendix

\input{appendix}

\end{document}

%% file: headers.tex
\usepackage[margin=1in]{geometry}
\usepackage{times}
\usepackage{amsmath,amssymb,amsfonts,amsthm}
\usepackage{mathtools}
\usepackage{graphicx}
\usepackage[colorlinks=true,allcolors=blue]{hyperref}
\usepackage[table,dvipsnames]{xcolor}
\usepackage{wrapfig}
\usepackage{tikz}
\usetikzlibrary{shapes, arrows, tikzmark, calc}
\usepackage{setspace}
\usepackage{nicefrac}
\usepackage[noend]{algpseudocode}
\usepackage[ruled,vlined,linesnumbered,algosection]{algorithm2e}
\usepackage[framemethod=tikz]{mdframed}
\usepackage{xspace, xparse}
\usepackage{pgfplots}
\usepackage{framed}
\usepackage{caption}
\usepackage{subcaption}
\usepackage{multicol}
\usepackage{changepage}
\usepackage[capitalise,noabbrev,nameinlink]{cleveref}
\usepackage[shortlabels,inline]{enumitem}
    \pgfplotsset{compat=1.5}
\usepackage{thmtools, thm-restate}
\usepackage{typed-checklist}
\usepackage{multirow}
\usepackage{booktabs}
\usepackage{makecell}
\usepackage{comment}
\usepackage{float}
\usepackage{url}
\usepackage{nicematrix}

\ifready
\else 
\usepackage{color-edits}
\addauthor[Felix]{fz}{cyan}
\fi

\usepackage[backref=true, natbib=true, backend=biber, style=alphabetic, sortcites=true, sorting=alphabeticlabel,
minbibnames=4, maxbibnames=999, mincitenames=999, maxcitenames=999, minalphanames=4, maxalphanames=4]{biblatex}
\DeclareSortingTemplate{alphabeticlabel}{
  \sort[final]{\field{labelalpha}} %
  \sort{
    \field{year}   %
  }
  \sort{
    \field{title}  %
  }
}
\AtBeginRefsection{\GenRefcontextData{sorting=ynt}}
\AtEveryCite{\localrefcontext[sorting=ynt]}

\definecolor{mydarkblue}{rgb}{0,0.08,0.45}
\hypersetup{ %
    colorlinks=true,
    linkcolor=mydarkblue,
    citecolor=mydarkblue,
    filecolor=mydarkblue,
    urlcolor=mydarkblue,
}

\newtheorem{theorem}{Theorem}[section]
\newtheorem{fact}[theorem]{Fact}
\newtheorem{assumption}[theorem]{Assumption}

\newtheorem{definition}[theorem]{Definition}
\newtheorem{lemma}[theorem]{Lemma}
\newtheorem{proposition}[theorem]{Proposition}

\newtheorem{corollary}[theorem]{Corollary}

\newtheorem{example}[theorem]{Example}
\newtheorem{remark}[theorem]{Remark}

\newcommand{\namedref}[2]{\hyperref[#2]{#1~\ref*{#2}}\xspace}

\usepackage{soul}

\NewDocumentEnvironment{pf}{o}
  {\IfNoValueTF{#1}{\begin{proof}}{\begin{proof}[Proof (#1)]}}
  {\IfNoValueTF{#1}{\end{proof}}{\end{proof}}}

%% file: macros.tex
\crefname{assumption}{assumption}{assumptions}
\Crefname{assumption}{Assumption}{Assumptions}
\crefname{step}{step}{steps}
\Crefname{step}{Step}{Steps}
\crefname{obstacle}{obstacle}{obstacles}
\Crefname{obstacle}{Obstacle}{Obstacles}
\crefname{change}{change}{changes}
\Crefname{change}{Change}{Changes}
\crefname{ineq}{inequality}{inequalities}
\Crefname{ineq}{Inequality}{Inequalities}

\newcommand{\storelines}{\xdef\rememberedlines{\number\value{AlgoLine}}}
\newcommand{\restorelines}{\setcounter{AlgoLine}{\rememberedlines}}

\newcommand{\Lap}{\ensuremath{{\textsf{Lap}}}\xspace}

\DeclareMathOperator*{\poly}{poly}

\DeclarePairedDelimiter{\abs}{|}{|}

  \newcommand{\eps}[0]{\ensuremath{\varepsilon}}
  \let\epsilon\eps

  \newcommand{\cD}{\ensuremath{{\mathcal D}}\xspace}

\mathchardef\hyph="2D

\newcommand{\alg}{\mathcal{A}}

\newcommand{\sset}{\subseteq}
\DeclarePairedDelimiter{\card}{\lvert}{\rvert}

\DeclarePairedDelimiter{\set}{\lbrace}{\rbrace}

\ifready
\else

\fi

\newcommand{\prob}{\mathbf{P}}
\newcommand{\expect}{\mathbb{E}}

\DeclareMathOperator{\Range}{Range}
\DeclareMathOperator{\OPT}{OPT}

\newcommand{\R}{\mathbb{R}}

\newcommand{\DP}{\ensuremath{\text{DP}}\xspace}

\newcommand{\mcal}{\mathcal}

\newcommand{\myparagraph}[1]{\paragraph{#1.}}

\newcommand{\eg}{\emph{e.g.}}
\newcommand{\ie}{\emph{i.e.}}

%% file: body.tex
\section{Introduction}

Differential privacy~\cite{DMNS06} (DP) has become the gold standard for privacy-preserving data analysis.
While originating in numerical data,
there is growing interest in analyzing more structured data
such as those represented by graphs~\cite{dinitz2025differentially,AU19,blocki2022privately,dinitz2024tight,DLRSSY22,ELRS22,kasiviswanathan2013analyzing,kalemaj2023node,
LUZ24,mueller2022sok,NRS07,RS16,raskhodnikova2016differentially,Upadhyay13}.
We focus on \emph{edge}-differential privacy, 
which requires output distributions to be ``close'' for graphs differing in one edge (edge-neighboring), 
modeling scenarios such as large social networks where the sensitive connections between individuals must be protected.

Today's graphs are highly dynamic, 
with new connections between individuals forming every minute.
Online streaming algorithms process edge updates while maintaining accurate statistics after each update~\cite{chen2023sublinear,halldorsson2016streaming,cormode2018approximating,GS24}.
However,
these outputs may inadvertently leak sensitive information over time.
To mitigate this,
the \emph{continual release} model~\cite{DworkNPR10,CSS11} of differential privacy requires that the entire vector of \emph{all outputs} from an online streaming algorithm
satisfy DP.

At the same time,
modern graphs are massive
and cannot be stored into memory.
The well-studied non-private sublinear space streaming model~\cite{assadi2022semi,AG11,EHW16,Esfandiari2018,feigenbaum2005graph,FMU22,HuangP19,henzinger1998computing,
muthukrishnan2005data,mcgregor2014graph,mcgregor2015densest,mcgregor2018simple}
addresses this by processing updates using space sublinear in the input size, 
ideally sublinear in the number of vertices (streaming),
or, more often, sublinear in the number of edges and near-linear in the number of vertices (semi-streaming).

We study the sublinear space continual release model,
introduced by {Epasto, Liu, Mukherjee, and Zhou~\cite{epasto2024power}},
which asks that a continual release algorithm release accurate solutions after each edge update
while using sublinear space.
In particular,
we design algorithms for the \emph{densest subgraph} (DSG) problem in the insertion-only setting under this model
with improved approximation and space guarantees.

Our main tool is subsampling,
a well-studied technique in both the DP community~\cite{balle2018tight,steinke2022composition,wang2020subsampled,dong2022gaussian,zhu2019poisson,balle2020},
as well as the sublinear space streaming community~\cite{ahn2009graph,mcgregor2014graph,mcgregor2018simple,ahn2012analyzing,bar2002reductions}.
Running a continual release algorithm on a subsampled stream should intuitively reduce privacy loss (\emph{privacy amplification}) and space usage (\emph{sparsification}), since the algorithm sees fewer edges.
Provided the subsampling error is sufficiently small, 
this simple approach appears to simultaneously achieve two desiderata!
However, 
to the best of our knowledge, 
the connection between privacy amplification and sparsification via subsampling
has not been formally investigated in the context of private graph algorithms.

Although \cite{epasto2024power} employ subsampling to design space-efficient continual release DSG algorithms,
they do not achieve \emph{meaningful} privacy amplification.
That is,
their subsampling scheme does not improve the privacy-utility tradeoff of their algorithm.
Roughly speaking,
their algorithm incurs $\Omega(\log^2 n)$ additive error while using $\Omega(n\log^2 n)$ space, 
compared to $O(\log n)$ error for the best static DP DSG algorithms~\cite{DLL23,dinitz2024tight} 
and $O(n\log n)$ space for the best non-private streaming algorithms~\cite{mcgregor2015densest,EHW16}.
Thus, their approach suffers an $\Omega(\log n)$ overhead in both error and space usage,
which can be impractical for modern large-scale graphs.
Moreover,
this is not an artifact of their analysis,
as we construct a simple example where their algorithm necessarily suffers the extra $\Omega(\log n)$ factor in error and space
(see \Cref{ex:no-amplification} in \Cref{sec:technical-overview}).

We make two simple conceptual changes to the \cite{epasto2024power} algorithm:
\begin{multicols}{2}
\begin{enumerate}[nosep]
    \item \emph{densification} of the initial graph and \label[change]{change:initial-graph-densification}
    \item modification of their subsampling scheme. \label[change]{change:subsampling-scheme-variant}
\end{enumerate}
\end{multicols}
These changes unlock privacy amplification via subsampling on top of sparsification,
providing the first formal connection between the two objectives,
and allowing us to improve the privacy/utility tradeoffs.
Our techniques removes the $\Omega(\log n)$ overhead in both additive error and space complexity from \cite{epasto2024power},
leading to the first continual release DSG algorithm with $O(\log n)$ additive error
and $O(n\log n)$ space.
In addition,
these changes allow us to simplify some of the analysis in \cite{epasto2024power},
which needed to adapt, for example, the \cite{mcgregor2015densest,EHW16} approximation proof for subsampling,
due to the existence of additive error necessary to satisfy DP.

Interestingly,
both changes can be interpreted as densifying the graph,
meaning we increase the number of edges while preserving some properties of the graph.
\Cref{change:initial-graph-densification} explicitly densifies the graph
while \Cref{change:subsampling-scheme-variant} can be interpreted as implicitly densifying the graph
before subsampling more aggressively
(see \Cref{sec:technical-overview:error}).
We believe this notion of graph densification may be of independent interest for other DP graph algorithms.

\subsection{Our Contributions}
Our main contributions are simple insertion-only continual release densest subgraph (DSG) algorithms
with improved additive error and space complexity.
After each update,
our algorithms release a differentially private vertex subset $S$ after each edge update
that induces an $(\alpha, \zeta)$-approximate DSG of the current graph.
Here $\alpha$ is the multiplicative (relative) error
and $\zeta$ denotes the additive error.

\begin{theorem}[Informal; See \Cref{thm:simple-dsg-formal}]\label{thm:simple-dsg-informal}
    Let $\eta\in (0, 1)$.
    There are $(\eps, \delta)$-DP densest subgraph algorithms in the insertion-only continual release model
    that achieve:
    \begin{enumerate}[nosep]
        \item For $\delta=0$, a $\left( 2+\eta, O\left( \frac{\log(n)\log(\nicefrac1\eta)}{\eps\eta} \right) \right)$-approximation
        using $O\left( \frac{n\log(n)\log(\nicefrac1\eta)}{\eps\eta^2} \right)$ space.
        \item For $\delta>0$, a $\left( 1+\eta, O\left( \frac{\log(\nicefrac{n}\eta)\log(\nicefrac1\eta)}{\eps\eta} \right) \right)$-approximation
        using $O\left( \frac{n\log(\nicefrac{n}\eta)\log(\nicefrac1\eta)}{\eps\eta^2} \right)$ space
        in the regime $\delta = \poly(\nicefrac1n)$.
    \end{enumerate}
\end{theorem}

In \Cref{table:results},
we contextualize our DSG algorithms.
Specifically,
we consider the current state-of-the-art continual release algorithms~\cite{epasto2024power}\footnote{We remark that \cite{epasto2024power} only explicitly treat pure-DP algorithms but the approximate-DP continual release algorithm in \Cref{table:results} is implicit in their work.},
the static DP algorithms~\cite{DLL23,dinitz2024tight} with lowest additive error,
as well as the streaming algorithms~\cite{mcgregor2015densest,EHW16} with best space complexity.

\begin{table}[htbp!]
\centering
\caption{Our results are highlighted in green.
$\eta\in (0, 1)$ is an approximate parameter.
For $(\eps, \delta)$-DP,
we assume $\delta = \poly(\nicefrac1n)$ for simplicity.
All results are for edge-DP. 
Bicriteria approximations $(\alpha, \zeta)$ are given, where $\alpha$ is the multiplicative factor and $\zeta$ is the additive error.
In red,
we highlight the logarithmic overhead in the additive error and space complexity from \cite{epasto2024power}
compared to the static DP and non-private streaming settings,
respectively.}
\label{table:results}
\resizebox{\textwidth}{!}{
\vspace{1em}
\begin{NiceTabular}{|c|c|c|c|c|}
\noalign{\global\arrayrulewidth=0.2mm}
\hline
Privacy & Approximation & Space & Reference & Model \\
\noalign{\global\arrayrulewidth=0.8mm} 
\hline
\noalign{\global\arrayrulewidth=0.2mm}
    $\eps$-DP & $\left( 2+\eta, O\left( \frac{\log(n)\log(\nicefrac1\eta)}{\eps\eta} \right) \right)$ & $O\left( \frac{n\log(n)\log(\nicefrac1\eta)}{\eps\eta^2} \right)$ & \cellcolor{YellowGreen!20}\Cref{thm:simple-dsg-informal} & \makecell[c]{Continual\\Release} \\ \hline
    $(\eps, \delta)$-DP & $\left( 1+\eta, O\left( \frac{\log(\nicefrac{n}\eta)\log(\nicefrac1\eta)}{\eps\eta} \right) \right)$ & $O\left( \frac{n\log(\nicefrac{n}\eta)\log(\nicefrac1\eta)}{\eps\eta^2} \right)$ & \cellcolor{YellowGreen!20}\Cref{thm:simple-dsg-informal} & \makecell[c]{Continual\\Release} \\ \hline\hline
    
    $\eps$-DP & $\left( 2+\eta, O\left( \frac{\log^{\color{red}2} n}{\eps\eta} \right) \right)$ & $O\left( \frac{n\log^{\color{red}2} n}{\eps\eta^2} \right)$ & \cite{epasto2024power} & \makecell[c]{Continual\\Release} \\ \hline
    $(\eps, \delta)$-DP & $\left( 1+\eta, O\left( \frac{\log^{\color{red}2} (\nicefrac{n}\eta)}{\eps\eta} \right) \right)$ & $O\left( \frac{n\log^{\color{red}2} (\nicefrac{n}\eta)}{\eps\eta^2} \right)$ & \cite{epasto2024power} & \makecell[c]{Continual\\Release} \\ \hline\hline
    
    $\eps$-DP & $\left( 2, O\left( \frac{\log n}\eps \right) \right)$ & $O\left( n+m \right)$ & \cite{DLL23} & Static DP \\ \hline
    $(\eps, \delta)$-DP & $\left( 1, O\left( \frac{\log n}\eps \right) \right)$ & $O\left( n+m \right)$ & \cite{dinitz2024tight} & Static DP \\ \hline\hline

    None & $\left( 1+\eta, 0 \right)$ & $O\left( \frac{n\log n}{\eta^2} \right)$ & \cite{mcgregor2015densest,EHW16} & Semi-Streaming \\ \hline
\end{NiceTabular}
}
\end{table}

The \cite{epasto2024power} framework
invokes a static DP DSG algorithm in a black-box manner, 
and its error and space bounds depend on the chosen static algorithm.
Their $\eps$-DP algorithm is obtained from the \cite{DLL23} algorithm
and their $(\eps, \delta)$-DP algorithm is obtained from the \cite{dinitz2024tight} algorithm.
The reduction of \cite{epasto2024power} roughly incurs a $O(\frac{1}{\eta}\log n)$ overhead in the additive error compared to the static algorithm.
Moreover, 
for technical reasons discussed in \Cref{sec:technical-overview}, 
their space usage depends on the additive error, 
leading to an additional $O(\frac1\eps \log n)$ space overhead relative to the non-private streaming setting.

We refine their approach and obtain \Cref{thm:simple-dsg-informal} via a simple black-box reduction to the static DP setting.
Our reduction improves the additive error overhead from $O(\frac1\eta\log n)$ to $O(\frac1\eta \log \frac1\eta)$ and the space overhead from $O(\frac1\eps\log n)$ to $O(\frac1\eps \log \frac1\eta)$.
Thus, we match the static setting’s additive error and the non-private semi-streaming space complexity up to constants.
Moreover,
by known lower bounds that hold even in the static $(\eps, \delta)$-DP setting~\cite{AHS21,NV21},
the additive error we achieve is tight up to a factor of $O(\sqrt{\frac1\eps \log n})$.

\subsubsection{Technical Contributions}
\myparagraph{Subsampling} 
The main tool that unlocks the improvement in both the approximation and space guarantees 
is a careful implementation of subsampling.
Although subsampling is well-studied in both the differential privacy and streaming literature, 
we are the first to formally establish its combined benefits in the graph privacy setting, 
yielding non-trivial improvements in both privacy and space.
We expect this connection to be of independent interest beyond our work.

\myparagraph{Graph Densification}
While the relationship between privacy amplification and sparsification via subsampling is intuitive,
the exact implementation of subsampling requires care.
As a concrete example,
the uniform subsampling algorithm for streaming DSG~\cite{mcgregor2015densest,EHW16} does not begin subsampling until the number of edges exceeds $\Omega(n\log n)$.
Hence for the first $O(n\log n)$ releases,
there is potentially no privacy amplification.
This partially explains the extra logarithmic factor in \cite{epasto2024power}.

We address this by introducing \emph{graph densification} for graph DP: 
artificially adding edges so that the initial graph has $\Omega(n\log n)$ edges, 
enabling privacy amplification from the start 
(at the cost of a more involved utility analysis).
Our modification of the subsampling scheme can also be understood as densification.
Graph densification remains underexplored even in the non-private setting~\cite{hardt2012densification}, 
and to the best of our knowledge, 
we are the first to formalize it for graph DP.
We believe it has potential applications to other problems in this area.

\section{Technical Overview}\label{sec:technical-overview}
In this section,
we provide a detailed overview of our subsampling implementation, 
which simultaneously improves the additive error via privacy amplification 
as well as the space complexity due to sparsification.
A key tool upon which we heavily rely to obtain our improved reductions is privacy amplification via subsampling.
\begin{theorem}[Theorem 8 in \cite{balle2018tight}]\label{thm:privacy-amplification-subsampling}
    Let $\eps, \delta\in [0, 1]$.
    and suppose $\mcal A(X)$ is an $(\eps, \delta)$-DP algorithm over a dataset $X$.
    If $Y(X)$ is obtained from $X$ by subsampling each element of $X$ independently with probability $q$,
    then $\mcal A(Y(X))$ is $(2q\eps, q\delta)$-DP.
\end{theorem}

Before presenting our techniques, we briefly review the \cite{epasto2024power} reduction.
Their algorithm privately tracks three quantities: 
the current number of edges $m_t$, 
the maximum density $\rho_t$, 
and a vertex subset $S_t$ inducing an (approximate) DSG.
They then perform lazy updates, 
releasing a new solution only when $\rho_t$ increases by a factor $(1+\eta)$ over the previous release.
To avoid storing the entire graph, they adapt the uniform subsampling method of \cite{mcgregor2015densest,EHW16}, and independently retain each edge with probability $q_t \approx \frac{n\log n}{m_t\eta^2}$.
The density value is computed on the sparsified graph $H_t$ 
and the actual density estimate is scaled by $\frac1{q_t}$,
yielding a $(1+\eta)$-multiplicative approximation of $\rho_t$.
Vertex subset solutions are obtained by running a static DP DSG algorithm on $H_t$.

One technical difficulty in the analysis of \cite{epasto2024power} is the presence of additive error due to privacy.
For example, 
$q_t$ depends on $m_t$, 
which is private and must be privately estimated by adding noise.
This additive error affects the subsampling rate, 
preventing a direct application of the \cite{mcgregor2015densest,EHW16} analysis 
and requiring a lengthy adaptation to handle the noise.
As we will see,
our approach sidesteps this: 
by densifying the initial graph so that $m_0$ is sufficiently large, 
we can treat additive error as $O(1)$ multiplicative error 
and reuse prior analyses in a largely black-box fashion.

\subsection{Error Overhead}\label{sec:technical-overview:error}
Recall that static DP DSG algorithms achieve $O(\frac1\eps \log n)$ additive error
in both the approximate and pure DP setting,
while \cite{epasto2024power} suffers an additive error of $\Omega(\frac{\log^2 n}{\eps \eta})$.
The $\Omega(\frac1\eta \log n)$ overhead results from a composition over $\log_{1+\eta}(n) = \Theta(\frac1\eta \log n)$ lazy updates,
as the density can increase by a factor of $(1+\eta)$ at most that many times.
This necessitates calling the static DP DSG algorithms with privacy parameter $\eps' \approx \frac{\eps\eta}{\log n}$
in order to ensure the final algorithm is $\eps$-DP.

We first describe an {ideal scenario} in which this overhead disappears.
Then we identify two obstacles in \cite{epasto2024power} where the ideal analysis fails.
Finally, we explain how our modifications overcome said obstacles.

\myparagraph{Ideal Scenario}
Suppose that the subsampling rate $q_t < 1$ 
so that privacy amplification via subsampling (\Cref{thm:privacy-amplification-subsampling}) holds.
We could run the static algorithm with parameter $\eps'$ for each lazy release while incurring only $O(q_t\eps')$ privacy loss.
If the number of edges doubled after every release,
$q_t$ reduces by a constant factor each time (\eg{}, $q_{t+1} \leq \nicefrac{q_t}2$).
The total privacy loss from computing vertex solutions would then be bounded by the geometric series
$\sum_{i\geq 0} \nicefrac{\eps'}{2^i} = O(\eps')$.
In other words,
it suffices to call the static DP DSG algorithms with privacy parameter $\eps' = O(\eps)$
in order to ensure the final algorithm is $\eps$-DP.
Two obstacles prevent this ideal behavior in \cite{epasto2024power}:
\begin{enumerate}[(O1)]
    \item The subsampling rate halves when the number of edges $m$ doubles,
    but the density can grow by a $(1+\eta)$-factor 
    at least $\Omega(\log n)$ times before $m$ doubles (see \Cref{ex:no-amplification}).\label[obstacle]{obstacle:edge-density-mismatch}
    \item There is no subsampling (and hence no privacy amplification) for the first $\Omega(\frac{n\log n}{\eta^2})$ edge updates as $m_t = O(\frac{n\log n}{\eta^2})$
    and so $q_t = 1$ during those time stamps.\label[obstacle]{obstacle:sparse-initial-graph}
\end{enumerate}

\begin{example}\label{ex:no-amplification}
    Let $\eta=1$.
    Consider a graph with $Cn^{\frac23}\log n$ disjoint $n^{\frac16}$-cliques plus a single $n^{\frac13}$-clique.
    The densest subgraph is the $n^{\frac13}$-clique, 
    and the total number of edges is $\Theta(n\log n)$.
    Choose $C$ so that the subsampling probability is $\frac12$.
    If we grow the $n^{\frac13}$-clique, 
    its density can double $\Omega(\log(n^{\
    \frac12-\frac13})) = \Omega(\log n)$ times before the number of edges doubles.
    Thus, 
    the \cite{epasto2024power} algorithm must release $\Omega(\log n)$ times before the subsampling rate $q_t$ decreases.
\end{example}

\myparagraph{\Cref{obstacle:edge-density-mismatch}} 
The first obstacle can be addressed with a conceptually simple solution.
In the hard instance from \Cref{ex:no-amplification}, 
the graph has many edges overall but remains sparse except for a small dense subgraph.
Thus, we can uniformly add edges to each vertex without significantly affecting the maximum density.
More generally, 
we add edges after each release so that every vertex has degree at least $\eta\rho_t$, 
where $\rho_t$ is the current maximum density.
This increases the maximum density by a $(1+\eta)$-factor, 
while ensuring $m_t \geq \Omega(\eta n\rho_t)$
so the subsampling rate halves after roughly $\nicefrac1\eta$ releases.
In general, $n\rho_t \le m_t$, but the reverse inequality need not hold.

While this approach achieves the guarantees of \Cref{thm:simple-dsg-informal}, it complicates the correctness and privacy proofs because the stream is adaptively lengthened.
Instead, we note that this operation is equivalent to more aggressive subsampling
with $q_t \approx \frac{\log(n)}{\rho_t \eta^2}$,
as we are essentially forcing $m_t\geq \eta n\rho_t$.
From examining the \cite{EHW16} analysis, 
this adjusted rate still preserves densities up to a small multiplicative factor
and computing a DSG on the sparsified graph and scaling by $\frac1{q_t}$ yields a $(1+\eta)$-approximation.
Note that $\rho_t\geq \frac{m_t}n$ so the new subsampling rate is at most the previous rate.

Let $t^\star$ be the first release after subsampling begins.
Because $q_t$ now depends on $\rho_t$ rather than $m_t$, 
it decreases by a factor $(1+\eta)$ after each release for $t \ge t^\star$, 
mostly resolving \Cref{obstacle:edge-density-mismatch}.
This also \emph{simplifies} the \cite{epasto2024power} algorithm: 
we now need to privately track only $\rho_t$ and $S_t$, not $m_t$.
As mentioned,
this can alternatively be understood as \emph{implicit} densification.

In summary, 
we can run a static DSG algorithm with privacy parameter $\eps' \approx \eps\eta$, 
so after subsampling starts ($t \ge t^\star$) the privacy loss from computing vertex solutions is
$
    \sum_{i\geq 0} \frac\eps\eta(1+\eta)^{-i}
    = O(\eps).
$
It remains to address the privacy loss \emph{before} the algorithm begins subsampling (\ie{}, for $t<t^\star$).

\myparagraph{\Cref{obstacle:sparse-initial-graph}} 
We address the second obstacle by \emph{explicitly} densifying the graph.
Specifically,
we artificially add edges from a $\kappa$-regular graph on $n$-vertices to the initial graph
for some $\kappa$ to be chosen later.
Intuitively,
this ensures that the algorithm begins subsampling much earlier.

Before determining $\kappa$,
one implementation detail we must address
is that the \cite{epasto2024power} framework does not support repeated edge insertions, 
yet densification can introduce them.
Since the algorithm does not store all edges, tracking $m_t$ with repeated edges may be problematic.
Now,
this would not be a problem if their subsampling rate $q_t$ does not explicitly depend on $m_t$,
as we can simply ignore the (non)-existence of previously processed edges.
Specifically,
when an edge $e_t$ is inserted, 
we delete previous copies in the sample (if any) 
and keep $e_t$ with probability $q_t$.
This yields a uniform subsample of the simple graph formed by all updates, 
regardless of repetitions.
Fortunately,
the subsampling rate $q_t\approx \frac{\log n}{\rho_t \eta^2}$ of our algorithm is a function of the current maximum density rather than $m_t$.
Hence we are indeed able to handle repeated edges with the scheme described above.

We now determine the densification degree $\kappa$.
Roughly speaking,
we set $\kappa \geq \Omega(\frac1\eta\log n)$.
This ensures that $\rho_0\geq \Omega(\kappa) = \Omega(\frac1\eta \log n)$ so that $q_0\approx \frac{\log n}{\rho_0\eta^2} = \nicefrac1\eta$.
Note that the density of any induced subgraph differs by at most $\kappa$ between the original and augmented graphs.
Recall from the previous paragraph that we execute the static algorithm with privacy parameter $\eps'\approx \eps\eta$,
hence the additive error $\Omega(\frac{\log n}{\eta\eps})$ already dominates $\kappa$.
Thus densification by a $\kappa$-regular graph only increases the additive error by a constant factor.
On the other hand, 
privacy amplification begins after at most $O(\log_{1+\eta}(\nicefrac1\eta))$ releases,
as the denominator in $q_t\approx \frac{\log n}{\rho_t\eta^2}$ can increase by a factor of $(1+\eta)$
at most that many times before $q_t < 1$.

All in all,
our algorithmic improvements
allow us to execute the static DSG algorithm with privacy parameter $\eps'\gets \frac{\eps\eta}{\log(\nicefrac1\eta)}$
so that the additive error overhead is now $O(\frac1\eta \log \frac1\eta)$ compared to the static DP setting,
rather than $O(\frac1\eta \log n)$ in \cite{epasto2024power}.

\subsection{Space Overhead}\label{sec:technical-overview:space}
Non-private streaming DSG algorithms~\cite{mcgregor2015densest,EHW16} achieve $O(\frac{n\log n}{\eta^2})$ space complexity,
while \cite{epasto2024power} requires $\Omega(\frac{n\log^2 n}{\eps \eta^2})$ space.
This $\Omega(\frac1\eps \log n)$ space overhead
is a consequence of the additive error of the static DSG algorithm.
Recall that their reduction executes a static DSG algorithm on the sparsified graph $H_t$,
whose density is roughly scaled by $\frac1{q_t}$ within $G_t$.
Thus if the static algorithm outputs a vertex subset inducing a subgraph with density $\frac{\OPT_{H_t}}{1+\alpha} - \zeta$ in $H_t$,
it induces a subgraph in $G_t$ with a density of approximately $\frac{\OPT_{G_t}}{1+\alpha} - \frac\zeta{q_t}$ in $G_t$
so the additive error can be amplified due to subsampling.
This was not an issue in the non-private streaming setting as it is possible to obtain pure multiplicative error algorithms.
In order to address this,
\cite{epasto2024power} also require $q_t\geq\frac{n\zeta}{m_t \eta}$,
so that the additive error can be absorbed into the relative error
\[
    \frac{\zeta}{q_t}
    \approx \zeta\cdot \frac{m_t \eta}{n\zeta}
    \leq \eta \frac{m_t}n
    \leq \eta\OPT\,.
\]
As mentioned in \Cref{sec:technical-overview:error},
$\zeta = \Omega(\frac{\log^2 n}{\eps \eta})$ in their framework since we need to execute the static algorithm with privacy parameter $\eps'\gets \frac{\eps\eta}{\log n}$ due to the lack of privacy amplification
and composition over $O(\frac1\eta \log n)$ releases.
But then $q_t\geq \frac{n\log^2 n}{m_t \eps \eta^2}$
so that in expectation,
the number of subsampled edges is $m_tq_t = \Omega(\frac{n\log^2 n}{\eps \eta^2})$,
which is $\Omega(\frac1\eps \log n)$ more than the non-private streaming setting.

Thus by reducing the additive error as in \Cref{sec:technical-overview:error}, 
we also reduce space. 
More specifically,
we now have $\zeta \approx O(\frac{\log n \log(\nicefrac1\eta)}{\eta \eps})$ so that the expected sample size is $O(\frac{n\log n \log(\nicefrac1\eta)}{\eps \eta^2})$,
which is indeed only $O(\frac1\eps\log\frac1\eta)$ more than the non-private setting.

\section{Related Works}\label{sec:related-works}
\myparagraph{Streaming Algorithms}
Streaming algorithms~\cite{morris1978counting,flajolet1985probabilistic,flajolet2007hyperloglog,alon1996space},
can efficiently process a stream of data in one pass without storing the full input.
Streaming graph algorithms ideally use space sublinear in the number of vertices~\cite{mcgregor2018simple,HuangP19},
but it is often necessary to allow for near-linear space in the number of vertices~\cite{assadi2022semi,AG11,Esfandiari2018,feigenbaum2005graph,FMU22,mcgregor2015densest}. 

Online streaming algorithms also use low space 
and additionally output solutions after 
every update~\cite{chen2023sublinear,halldorsson2016streaming,cormode2018approximating,GS24}.
In the absence of privacy constraints, 
there are many online streaming graph algorithms, 
including densest subgraph~\cite{EHW16,mcgregor2015densest}, 
$k$-core decomposition~\cite{Esfandiari2018,KingTY23},
and more~\cite{assadi2019coresets,AG11,ALT21,Assadi22,FMU22,mcgregor2018simple,BerenbrinkKM14,assadi2022semi,AG11,feigenbaum2005graph,FMU22,HuangP19,henzinger1998computing,
muthukrishnan2005data,mcgregor2014graph}

\myparagraph{Continual Release Model}
In the context of online streaming computation, the DP model of reference is the continual release model~\cite{DworkNPR10,CSS11},
This area has received significant attention
(see \eg{}, \cite{ChanLSX12,henzinger2024unifying, henzinger2023differentially,fichtenberger2023constant,jain2023counting,JRSS21,CSS11,henzinger2024unifying,henzinger2023differentially,ChanLSX12,epasto2023differentially}),
including in graphs~\cite{epasto2024power,FichtenbergerHO21,Song18,jain2024time,raskhodnikova2024fully}.
 
\myparagraph{Private Graph Algorithms}
There is extensive work on static DP graph algorithms (see \eg{}\ \cite{AU19,blocki2022privately,dinitz2024tight,DLRSSY22,ELRS22,kasiviswanathan2013analyzing,kalemaj2023node,
LUZ24,mueller2022sok,NRS07,RS16,raskhodnikova2016differentially,Upadhyay13,AU19,chen2023private,bun2021differentially,imola2023differentially} and references therein).
We focus on the densest subgraph problem,
for which there has been a variety of works~\cite{NV21,AHS21,DLL23,DLRSSY22,dinitz2025differentially}.
DSG has been extensively studied in the non-private streaming context. 
Our work builds on the DP results of~\cite{epasto2024power}, 
which in turn was inspired by the non-private work of~\cite{EHW16,mcgregor2015densest},
who show that uniformly subsampling suffices to approximately preserve density.

In the context of DSG, 
beyond the work of~\cite{epasto2024power}, 
all other works either output value-only~\cite{FichtenbergerHO21} 
or are in the static DP setting~\cite{DLL23,DLRSSY22,DLRSSY22}.
The best static DP DSG algorithms attain $O(\frac1\eps \log n)$ additive error.
Comparatively, 
the best-known lower bound is $\Omega(\sqrt{\frac1\eps\log n})$~\cite{AHS21,NV21},
which holds even for $(\eps,\delta)$-DP algorithms. 

Our paper presents the first continual release DSG algorithm for edge-private
insertion-only streams that match the additive error of the static DP algorithms
as well as the space complexity of the non-private streaming counterparts.
We remark that all previous algorithms,
even those that only report the density value,
incur greater additive error and use more space.

\section{Preliminaries}
We now introduce the notation we use in this paper as well as some definitions. 
Further standard preliminaries are deferred to \Cref{apx:prelims}.

\subsection{Setting \& Notation}
A graph $G=(V,E)$ consists of a set of vertices $V$ and a set of edges $E$ where edge $\{u,v\} \in E$ if and only if there is 
an edge between $u \in V$ and $v \in V$. We write $n:= \vert V \vert$ and $m:= \vert E \vert$. 

We consider the insertion-only setting within the continual release model,
where we begin with a (possibly empty) graph $G_0 = (V, E_0)$ where $E_0\sset \binom{V}{2}$
and edge updates in the form of $\{e_t, \perp\}$ arrive in a stream at every timestep $t \in [T]$.
Importantly,
we allow the same edge to be inserted multiple times
but the graph of interest is always the simple graph given by the set of inserted edges.
We are required to release an output for the problem of interest after every {(possibly empty)} update. 
Our goal is to obtain algorithms that achieve sublinear space in the number of edges $\tilde{o}(m)$ (note that $T \geq m$ as we allow for empty updates). 
We assume $T\leq n^{\bar c}$ for some absolute constant $\bar c\geq 1$ to simplify our presentation. 
If there are no empty edge updates or repeated edges,
it is sufficient to take $\bar c=2$.

Throughout the paper,
we write $\eta, \alpha$ to denote multiplicative approximation parameters,
and $\zeta$ to denote the additive error.

\subsection{Differential Privacy \& Continual Release}
We use the phrasing of the below definitions as given in~\cite{jain2024time}.

\begin{definition}[Graph Stream~\cite{jain2024time}]\label{def:graph-stream}
    In the continual release model, 
    an insertion-only graph stream $S \in \mathcal{S}^T$ of length $T$ is a $T$-element vector
    where the $t$-th element is an edge update $u_t\in\set{vw, \bot}$ indicating the insertion of the edge $vw$,
    or an empty operation. 
\end{definition}

We use $G_t$ and $E_t$ to denote the simple graph induced by the set of updates in the stream $S$ up to and including update $t$.
Now, we define neighboring streams as follows. Intuitively,
two graph streams are edge neighbors if one can be obtained from the other by removing one edge update 
(replacing the edge update by an empty update in a single timestep). 

\begin{definition}[Edge Neighboring Streams]\label{def:neighboring-streams}
    \sloppy
    Two streams of edge updates, $S = [u_1, \dots, u_T]$ and $S' = [u'_1, \dots, u'_T]$, 
    are \emph{edge-neighboring} if there is exactly one $t^\star\in [T]$ such that $u_t = u_t'$ for all $t\neq t^\star$.
\end{definition}

The notion of neighboring streams leads to the following definition of an edge differentially private algorithm.
\begin{definition}[Edge Differential Privacy]\label{def:edge DP}
  Let $\eps, \delta\in (0, 1)$.
  A continual release algorithm $\mcal A(S): \mathcal{S}^T \rightarrow \mathcal{Y}^T$ that takes as input a graph stream $S \in \mathcal{S}^T$
  is said to be \emph{$(\eps, \delta)$-edge differentially private (DP)}
  if for any pair of edge-neighboring graph streams $S, S'$ that differ by 1 edge update
  and for every set of outcomes $Y\sset \Range(\mcal A)$,
  \[
    \prob\left[ \mcal A(S)\in Y \right]
    \leq e^\eps \cdot \prob\left[ \mcal A(S')\in Y \right] + \delta.
  \]
  When $\delta=0$,
  we say that $\mcal A$ is \emph{$\eps$-edge DP}.
\end{definition}

\subsection{Differential Privacy Tools} \label{prelim:dp}
Throughout the paper,
we use some standard privacy mechanisms as building blocks (see~\cite{dwork2014algorithmic} for a reference).

\begin{definition}[Global sensitivity]
The global sensitivity of a function $f:\cD\to\mathbb{R}^d$ is defined by
\[\Delta_f=\max_{D,D'\in\cD,D \sim D'}\|f(D)-f(D')\|_1.\]
where $D \sim D'$ are neighboring datasets and differ by an element. 
\end{definition}

\begin{definition}[Laplace Distribution]
We say a random variable $X$ is drawn from a Laplace distribution with mean $\mu$ and scale $b>0$ if the probability density function of $X$ at $x$ is $\frac{1}{2b}\exp\left(-\frac{|x-\mu|}{b}\right)$. 
We use the notation $X\sim\Lap(b)$ to denote that $X$ is drawn from the Laplace distribution with scale $b$ and mean $\mu=0$. 
\end{definition}
The Laplace mechanism for $f: X\to \R$ with global sensitivity $\sigma$ 
adds Laplace noise to the output of $f$ with scale $b=\nicefrac\sigma\eps$
before releasing.
\begin{proposition}[\cite{dwork2014algorithmic}]
    The Laplace mechanism is $\varepsilon$-DP.
\end{proposition}

\begin{theorem}[Adaptive Composition; \cite{DMNS06,DL09,DRV10}]\label{thm:composition}
    A sequence of DP algorithms, $(\alg_1, \dots, \alg_k)$, with privacy parameters $(\eps_1, \delta_1), \dots, (\eps_k, \delta_k)$ form at worst an $\left(\eps_1 + \cdots + \eps_k, \delta_1 + \dots + \delta_k \right)$-DP algorithm under \emph{adaptive composition} (where the adversary can adaptively select algorithms after
    seeing the output of previous algorithms).%
\end{theorem}

\subsubsection{Sparse Vector Technique}
\SetKwProg{Class}{Class}{}{}
\SetKwProg{Fn}{Function}{}{end}
\SetKwFunction{SVT}{SVT}
\SetKwFunction{ProcessQuery}{ProcessQuery}
Below, we define the \emph{sparse vector technique} and give its privacy and approximation guarantees. The sparse
vector technique is used to privately answer \emph{above threshold} queries where an above threshold query checks whether the 
output of a function that operates on an input graph $G$ exceeds a threshold $T$.

Let $D$ be an arbitrary (graph) dataset,
$(f_t, \tau_t)$ a sequence of (possibly adaptive) query-threshold pairs,
$\Delta$ an upper bound on the maximum sensitivity\footnote{Recall the \emph{sensitivity} of a real-valued query is the maximum absolute difference between the values of the query on neighboring datasets.} of all queries $f_t$.
The algorithm stops running at the first instance of 
the input exceeding the threshold
and also outputs a differentially private estimate of the query value.

We use the variant introduced by {Lyu, Su, and Li~\cite{lyu2017understanding}},
detailed in the class \SVT{$\eps, \Delta$} (\Cref{alg:sparse vector technique})
where $\eps$ is our privacy parameter
and $\Delta$ is an upper bound on the maximum sensitivity of incoming queries.
The class provides a \ProcessQuery{$query, threshold$} function
where $query$ is the query to \SVT and $threshold$ is the threshold
that we wish to check whether the query exceeds.
This variant also outputs an estimate of the value of the query along with the output ``above''.

\begin{theorem}[Theorem 4 in \cite{lyu2017understanding}]\label{thm:sparse vector technique}
    \Cref{alg:sparse vector technique} is $\eps$-DP.
\end{theorem}

\begin{algorithm}[htp!]
\caption{Sparse Vector Technique \cite[Algorithm 7]{lyu2017understanding}}\label{alg:sparse vector technique}
Input: privacy budget $\eps$, upper bound on query sensitivity $\Delta$\\
\Class{\SVT{$\eps, \Delta$}}{
    $\xi \gets \Lap(\nicefrac{3\Delta}{\eps})$ \label{svt:threshold noise} \\

    \Fn{\ProcessQuery{$f_t(D), \tau_t$}}{
        \If{$f_t(D) + \Lap({6\Delta}/{\eps}) \geq \tau_t + \xi$}{ \label{svt:query noise}
          $\tilde f_t \gets f_t(D) + \Lap(\nicefrac{3\Delta}{\eps})$ \\
          \Return ``above'', $\tilde f_t$, and Abort \\
        }
        \Else{
          \Return ``below'' \\
        }
    }
}
\end{algorithm}

\section{Algorithm}
\SetKwFunction{PrivateDensestSubgraph}{PrivateDensestSubgraph}
We now present our framework,
which relies a static DP densest subgraph algorithm.
We formalize this assumption below.
\begin{assumption}\label{assum:static-dp-dsg-alg}
    Let $\alpha\in[0, 1]$.
    We assume the existence of a static $(\eps, \delta)$-DP algorithm \PrivateDensestSubgraph 
    such that given an $n$-vertex $m$-edge graph and desired failure probability $n^{-c}$,
    there is an absolute constant $C > 0$ such that
    the algorithm outputs a $\left( 1+\alpha, C\cdot \zeta(n, \eps, \delta, \alpha) \right)$-approximate densest subgraph
    with probability $1-n^{-c}$
    while using $O(m+n)$ space.
\end{assumption}
We note that previous works satisfy \Cref{assum:static-dp-dsg-alg} with $\alpha\in [0, 1]$ and $\zeta = \poly(\log(n), \nicefrac1\eps, \log(\nicefrac1\delta), \nicefrac1\alpha)$.
See \Cref{apx:static-dp-dsg-alg}.

\myparagraph{Pseudocode} 
The pseudocode of our algorithm can be found at \Cref{alg:simple-dsg-struct,alg:simple-dsg-main}.
\Cref{alg:simple-dsg-struct} describes the main data structure of interest
while \Cref{alg:simple-dsg-main} is a simple wrapper around \Cref{alg:simple-dsg-struct}
that calls the appropriate functions upon an edge update.
The main conceptual changes from \cite{epasto2024power} are highlighted in blue.

\myparagraph{Notation}
For the sake of convenience,
we summarize the specific notation used in \Cref{alg:simple-dsg-struct,alg:simple-dsg-main}.
\begin{multicols}{2}
\begin{description}[nosep]
    \item[{$\eps, \delta\in [0, 1]$}] privacy parameters
    \item[$\eta\in (0, 1)$] multiplicative approximation factor
    \item[$n$] given number of vertices in the graph
    \item[$c\geq 1$] exponent in the failure probability $n^{-c}$
    \item[$C>0$] absolute constant depending only on $c$
    \item[$\Upsilon > 0$] used to adjust the parameter of subroutines
    \item[$\kappa > 0$] initial graph densification parameter
    \item[$F = F_t\sset E_t$] subsampled edge set 
    \item[$H = H_t$] sparsified graph, where we abuse notation to also indicate a hashmap
    \item[$\rho_\DP = \rho_\DP^{(t)}$] DP estimate of the maximum density $\rho_t$
    \item[$S_\DP = S_\DP^{(t)}\sset V$] DP vertex subset solution inducing an approximate densest subgraph
    \item[$r = r_t$] true maximum density of $H_t$
    \item[$\tilde r = \tilde r_t$] DP estimate of $r_t$
\end{description}
\end{multicols}

\myparagraph{Algorithm Description}
Similar to the \cite{epasto2024power} framework,
our reduction performs lazy updates,
releasing a new solution only when
an instance of the sparse vector technique (\Cref{alg:sparse vector technique}) approximately detects a $(1+2\eta)$-factor increase in the maximum density value.
We also subsample the input stream by deciding whether to keep an incoming edge with some probability.
This probability is gradually decreased as more edges are inserted.

As mentioned in \Cref{sec:technical-overview},
we make two conceptual changes to the \cite{epasto2024power} framework.
First,
we adjust the subsampling rate $q_t\approx \frac{\log n}{\rho_t \eta^2}$
to explicitly depend only on $\rho_t$,
whereas \cite{epasto2024power} depend on the number of edges $m_t$.
This addresses the potential issue described in \Cref{obstacle:edge-density-mismatch} from \Cref{sec:technical-overview},
where the \cite{epasto2024power} framework is forced to release many solutions
before the subsampling rate is reduced.
Specifically,
note that the update to the subsampling probability $q$ (\Cref{alg:simple-dsg-struct:subsample-prob-update})
decreases by a factor of at least $(1+2\eta)$
when the density estimate increases by a factor of at least $(1+2\eta)$ (\Cref{alg:simple-dsg-struct:density-threshold-update}).
This step ensures that the privacy loss of each solution release is geometrically decreasing due to the geometrically decreasing subsampling probability (\Cref{thm:privacy-amplification-subsampling}).

Second,
we begin the algorithm by artificially adding edges from a $2\kappa$-regular graph 
for $\kappa \approx \log n$.
This incurs an additive error of approximately $\kappa$
since every induced subgraph on $k$ vertices changes by roughly $k\cdot\kappa$.
In exchange,
we ensure that the density is at least $\kappa\geq \log n$ throughout the duration of the algorithm
so that we begin subsampling much sooner than \cite{epasto2024power},
addressing \Cref{obstacle:sparse-initial-graph} from \Cref{sec:technical-overview}.

We highlight these changes in blue within \Cref{alg:simple-dsg-struct}.

\SetKwFunction{FRecurs}{FnRecursive}%
\SetKwFunction{PrivateContinualDSG}{PrivateContinualDSG}
\SetKwFunction{ProcessEdge}{ProcessEdge}
\SetKwFunction{UpdateSample}{UpdateSample}
\SetKwFunction{SampleEdge}{SampleEdge}
\SetKwFunction{GetNonPrivateDensity}{GetNonPrivateDensity}
\SetKwFunction{GetPrivateApproxDensity}{GetPrivateApproxDensity}
\SetKwFunction{GetPrivateApproxDensestSubgraph}{GetPrivateApproxDensestSubgraph}
\SetKwFunction{ExactDensity}{ExactDensity}
\SetKwFunction{PrivateDensestSubgraph}{PrivateDensestSubgraph}
\SetKwData{Svt}{svt}
\SetKwData{DensitySvt}{densitySvt}
\SetKwData{EdgeCounter}{edgeCounter}
\SetKwData{SubsampleTimer}{subsampleTimer}
\SetKwData{Count}{count}

\newcommand{\ConstPrivCount}{c_1}
\newcommand{\ConstDensitySvt}{c_2}
\newcommand{\ConstSubsampleProb}{c_3}

\begin{algorithm}[htp]
    \caption{Data Structure for Densest Subgraph in Adaptive Insertion-Only Streams\label{alg:simple-dsg-struct} }
    \Class{\PrivateContinualDSG{$\eps$, $\delta$, $\eta$, $n$, $c$}}{
        $C\gets$ sufficiently large constant depending only on $c$ \\
        {$\Upsilon\gets \log_{1+2\eta}(\nicefrac{3}\eta) + \frac{1+2\eta}{\eta}$} \\
        {$\kappa \gets \max\left( 
            C\cdot \zeta(n, \nicefrac\eps{\Upsilon}, \nicefrac\delta{\Upsilon}, \eta), 
            \frac{2C\Upsilon \ln(n)}{\eps}
            \right)$
        } \\
        {\color{blue} Initialize edge sample $F \gets$ edges of an arbitrary $n$-vertex $2\kappa$-regular graph} \label{one-shot:sampled}\\
        {\color{blue} Initialize hashmap $H[e] = 1$ for $e\in F$} \label{one-shot:sample-hashmap} \\
        {\color{blue} Initialize $\rho_\DP \gets \kappa$} \\
        {Initialize $S_\DP\gets V$} \\
        {Initialize sampling probability $q\gets 1$} \\
        Initialize class \DensitySvt$\gets$\SVT{$\nicefrac{\eps}{\Upsilon}, \Delta=1$} 
        \qquad(\Cref{alg:sparse vector technique})\label{one-shot:svt-func} \\
        \BlankLine
        \Fn{\SampleEdge{$e$}}{
            \If{$e\neq \perp$} {
                Remove $e$ from $F, H$ if stored\\
                Sample $h_{e} \sim U[0, 1]$ uniformly at random in $[0, 1]$.\label{one-shot:sample-h-t-value}\\
                \If{$h_{e} \leq q$}{\label{one-shot:less-than-prob}
                    Store $F \gets F \cup \{e\}$.\label{one-shot:sample-edge}\\
                    Set $H[e] = h_{e}$.\label{one-shot:store-sample}
                }
            }
        }
        \BlankLine
        \Fn{\UpdateSample{}}{
            \For{each edge $e \in H$}{\label{one-shot:iterate-stored-edges}
                \If{$H[e] > q$}{\label{one-shot:check-ht}
                    Remove $e$ from $F$ and $H$.\label{one-shot:remove-edge} 
                    \label{one-shot:keep-edge}
                } 
                    
            }
        }
        \BlankLine
        \Fn{\ProcessEdge{$e_t$}}{
            \SampleEdge{$e_t$} \\
            $r \gets$ exact density of densest subgraph in $(V, F)$ \qquad(\Cref{thm:densest value}) \\
            \If{\DensitySvt.\ProcessQuery{$r, q\cdot (1+2\eta)\rho_\DP$} returns ``above'', $\tilde r$}{
                $\rho_\DP\gets \max\set{(1+2\eta)\rho_\DP, \nicefrac{\tilde r}q}$
                \label{alg:simple-dsg-struct:density-threshold-update} \\
                Reinitialize class \DensitySvt$\gets$\SVT{$\nicefrac{\eps}{\Upsilon}, \Delta=1$} \\
                $S_\DP \gets$\PrivateDensestSubgraph{$F, \nicefrac\eps{\Upsilon}, \nicefrac\delta\Upsilon$} \qquad(\Cref{assum:static-dp-dsg-alg}) 
                \label{one-shot:get-new-private-graph}
                \label{one-shot:return-private-ds} \\
                {\color{blue}$q\gets \min\left( 1, \frac{3\kappa}{\rho_\DP \eta} \right)$}
                \label{alg:simple-dsg-struct:subsample-prob-update} \\
                \UpdateSample{}
            }
        }
	}
    \storelines
\end{algorithm}

\SetKwData{PrivateContinualDsg}{privateContinualDSG}

{
\begin{algorithm}[htp]
    \caption{Algorithm for Densest Subgraph in Adaptive Insertion-Only Streams \label{alg:simple-dsg-main}}
    \restorelines
    Initialize \PrivateContinualDsg$\gets$\PrivateContinualDSG{$\eps, \eta$, $n$, $c$}.\label{one-shot:initialize-dsg} \\
    \For{each new update $e_t$}{\label{one-shot:for}
       \PrivateContinualDsg.\ProcessEdge{$e_t$}.\label{one-shot:sample}\\ 
        Release \PrivateContinualDsg.$S_\DP$, \PrivateContinualDsg.$\rho_\DP$ \label{one-shot:release-stored} 
    }
\end{algorithm}
}

\section{Analysis}
In this section,
we prove the following formalization of \Cref{thm:simple-dsg-informal}
\begin{theorem}\label{thm:simple-dsg-formal}
    Let $\eta, \eps\in (0, 1)$ and $\delta, \alpha\in [0, 1]$.
    Suppose \Cref{assum:static-dp-dsg-alg} holds.
    There is an $(O(\eps), \delta)$-DP algorithm for densest subgraph in the continual release model
    with the following guarantees.
    With probability $1-\poly(\nicefrac1n)$,
    the algorithm outputs a vertex set inducing a $\left( (1+\eta)(1+\alpha), \kappa \right)$-approximate densest subgraph at every $t\in [T]$
    while using $O(\frac{n\kappa}\eta)$ space
    for
    \begin{align*}
        \kappa = O\left( \zeta\left( n, \frac\eps{\Upsilon}, \frac\delta\Upsilon, \alpha \right) + \frac{\Upsilon\ln(n)}{\eps} \right)\,, \qquad\qquad
        \Upsilon = O\left( \frac1\eta \log\frac1\eta \right) \,.
    \end{align*}
\end{theorem}
When instantiated with the appropriate static DP densest subgraph algorithms~\cite{DLL23,DLRSSY22,dinitz2025differentially} (\Cref{thm:private static densest subgraph,thm:improved private static densest subgraph,thm:static-approx-dp-dsg-alg}),
\Cref{thm:simple-dsg-formal} achieves the following results,
which is a superset of our informal \Cref{thm:simple-dsg-informal}.
\begin{description}[nosep]
    \item[\cite{DLL23}] An $\eps$-DP $\left( 2+\eta, O\left( \frac{\log(n)\log(\nicefrac1\eta)}{\eps\eta} \right) \right)$-approximation using $O\left( \frac{n\log(n)\log(\nicefrac1\eta)}{\eps\eta^2} \right)$ space.
    \item[\cite{DLRSSY22}] An $\eps$-DP $\left( 1+\eta, O\left( \frac{\log^4(n)\log(\nicefrac1\eta)}{\eps\eta^4} \right) \right)$-approximation using $O\left( \frac{n\log^4(n)\log(\nicefrac1\eta)}{\eps\eta^4} \right)$ space.
    \item[\cite{dinitz2024tight}] An $(\eps, \delta)$-DP $\left( 1+\eta, O\left( \frac{\log(\nicefrac{n}\eta)\log(\nicefrac1\eta)}{\eps\eta} \right) \right)$-approximation
using $O\left( \frac{n\log(\nicefrac{n}\eta)\log(\nicefrac1\eta)}{\eps\eta^2} \right)$ space
in the regime $\delta = \poly(\nicefrac1n)$.
\end{description}

In particular,
this yields a logarithmic improvement in the overhead of the reduction in both the additive error and space complexity.
Specifically,
the additive error overhead compared to the static algorithm is only a constant $O(\frac1\eta\log\frac1\eta)$
compared to the static DP setting
as opposed to the logarithmic $O(\frac1\eta\log(n))$ overhead from \cite{epasto2024power}.
Similarly,
we improve space overhead compared to the non-private streaming setting to a constant $O(\frac1\eps\log\frac1\eta)$ 
as opposed to the logarithmic $O(\frac1\eps\log n)$ overhead from \cite{epasto2024power}.

The proof of \Cref{thm:simple-dsg-formal} is completed in three parts.
We begin in \Cref{sec:simple-dsg-privacy}
by proving the privacy guarantees of \Cref{alg:simple-dsg-main} (\Cref{thm:simple-dsg-privacy}).
Then,
we analyze its approximation guarantees in \Cref{sec:simple-dsg-approximation} (\Cref{thm:simple-dsg-approximation}).
Last but not least,
\Cref{sec:simple-dsg-space} proves the space complexity of our algorithm (\Cref{thm:simple-dsg-space}).

\subsection{Privacy}\label{sec:simple-dsg-privacy}
We begin by analyzing the privacy guarantees of our framework.
\begin{theorem}\label{thm:simple-dsg-privacy}
    \Cref{alg:simple-dsg-main} is $(2\eps, \delta)$-DP.
\end{theorem}

We first provide a deterministic bound on the rate of decrease of the subsampling rate,
which is important to determine worst-case privacy loss
under privacy amplification via subsampling.
\begin{lemma}\label{lem:simple-dsg-subsampling-q-upper-bound}
    Let $0=t_0 < t_1 < t_2 < \dots$ be the timesteps where $\rho_\DP^{(t)}$ increased,
    \ie{}, when \DensitySvt answered ``above''.
    Then for all $t_{i-1} < t \leq t_{i}$,
    $F_t$ is independently subsampled from the edge set of the graph $G_t$ 
    with uniform probability
    \[
        q_t = q_{t_{i}}\leq \min\left( 1, \frac{3}{\eta(1+2\eta)^{i-1}} \right)\,.
    \]
\end{lemma}

\begin{proof}
    We argue by induction.
    The subsampling probability is initially set to $q^{(0)} = 1$.
    Suppose inductively that the statement holds up to index $i$.
    \Cref{alg:simple-dsg-main} sets $\rho_\DP^{(t_i)}\geq (1+2\eta) \rho_\DP^{(t_i-1)}$.
    Thus the second argument in the definition of $q^{(i+1)} = \max\left( 1, \frac{3\kappa}{\rho_\DP^{(t_i)}} \right)$ decreased by a factor of at least $(1+2\eta)$,
    as desired.
\end{proof}

We are now ready to prove \Cref{thm:simple-dsg-privacy}.
Compared to \cite{epasto2024power},
we also explicitly handle $(\eps, \delta)$-DP algorithms,
necessitating a more careful analysis through the privacy loss variable.

Roughly speaking,
the privacy guarantees of the sparse vector technique (\Cref{thm:sparse vector technique})
ensure that we only lose privacy when there is sufficient multiplicative increase in the optimal density
and we release a solution.
Since we proved a geometrically decreasing subsampling rate after each release (\Cref{lem:simple-dsg-subsampling-q-upper-bound}),
the privacy loss is geometrically decreasing as the algorithm proceeds.
More specifically,
conditioning on the transcript up to $t_i$ fixes sampling the subsampling probability $q_{t_i}$ identically on neighboring streams since $q_{t_i}$ is a function of the released $\rho_{\DP}^{(t_i)}$.
We can then apply standard subsampling amplification within the epoch (\Cref{thm:privacy-amplification-subsampling}) to deduce the geometrically decreasing privacy loss.

\begin{proof}[Proof of \Cref{thm:simple-dsg-privacy}]
    Let $(\rho_t, S_t)_{t=1}^{T}$ be the random variables output by \Cref{alg:simple-dsg-main},
    where $\rho_t\in \R$ is an estimate of the maximum density induced by \DensitySvt,
    and $S_t\sset V$ induces an approximate densest subgraph output by a static DSG algorithm.
    We write $(\rho_t', S_t')$ to be the random variables output by the algorithm on an edge neighboring stream.

    The privacy loss variable is defined by drawing outputs $(d_t, U_t)\sim (\rho_t, S_t)$,
    and computing
    \begin{align*}
        \mcal L
        &\coloneqq \ln\left( \frac{\Pr[\forall t\in [T], \rho_t=d_t, S_t=U_t]}{\Pr[\forall t\in [T], \rho_t'=d_t, S_t'=U_t]} \right) \\
        &= \sum_{t=1}^T \ln\left( \frac{\Pr[S_t=U_t\mid \rho_t=d_t, <t]}{\Pr[S_t'=U_t\mid \rho_t'=d_t, <t]} \right)
        + \ln\left( \frac{\Pr[\rho_t=d_t \mid <t]}{\Pr[\rho_t'=d_t \mid <t]} \right)\,.
    \end{align*}
    Here we used the shorthand $<t$ to denote the event that the output up $t-1$ is $(d_\tau, S_\tau)_{\tau=1}^{t-1}$.
    Recall that an algorithm is $\eps$-DP if and only if the absolute value of the privacy loss variable 
    is bounded above by $\eps$ with probability 1,
    whereas an algorithm is $(\eps, \delta)$-DP if and only if the absolute value of the privacy loss variable
    is bounded above by $\eps$ with probability $\delta$ over the draws $(d_t, U_t)$~\cite[Lemma 3.17]{dwork2014algorithmic}.
    We will upper bound $\mcal L$ as the lower bound follows by symmetry.

    We first consider the second term.
    Let $t_1 < t_2 < \dots$ be the timesteps where $\rho_{t_i}$ increased,
    \ie{}, when \DensitySvt answered ``above''.
    We can write
    \begin{align*}
        \sum_{t=1}^T \ln\left( \frac{\Pr[\rho_t=d_t \mid <t]}{\Pr[\rho_t'=d_t \mid <t]} \right)
        &=\sum_i \sum_{t=t_{i-1}+1}^{t_i} \ln\left( \frac{\Pr[\rho_t=d_t \mid <t]}{\Pr[\rho_t'=d_t \mid <t]} \right)\,.
    \end{align*}
    Thus, by viewing \DensitySvt as a simple composition over single-response SVT mechanisms,
    we see that the $i$-th inner summation satisfies
    \[
        \sum_{t=t_{i-1}+1}^{t_i} \ln\left( \frac{\Pr[\rho_t=d_t \mid <t]}{\Pr[\rho_t'=d_t \mid <t]} \right)
        \leq \eps_i,
        \qquad\text{where}
        \quad\eps_i \leq 
        \begin{cases}
            \frac{\eps}{\Upsilon}, &q_{t_i}=1\,, \\
            2q_{t_i}\frac{\eps}{\Upsilon}, &q_{t_i}<1\,,
        \end{cases}
    \]
    is guaranteed by SVT (\Cref{thm:sparse vector technique})
    and privacy amplification via subsampling (\Cref{thm:privacy-amplification-subsampling}).
    As noticed by \cite{epasto2024power},
    a subtle but important detail in this stage of the proof is that,
    conditioned on past events,
    we can pretend $q_{t_i}$ is identical on neighboring streams.
    Hence the standard analysis of privacy amplification via subsampling applies.
    
    The subsampling probability $q_{t_i}$ satisfies $q_{t_i}\leq \min\left( 1, \frac{3}{\eta(1+2\eta)^{i-1}} \right)$
    by \Cref{lem:simple-dsg-subsampling-q-upper-bound}.
    Hence with probability 1,
    \begin{align*} 
        \sum_{t=1}^T \ln\left( \frac{\Pr[\rho_t=d_t \mid <t]}{\Pr[\rho_t'=d_t \mid <t]} \right)
        &\leq \log_{1+2\eta}(\nicefrac{3}\eta)\cdot \frac{\eps}{\Upsilon} 
        + \sum_{j=1}^\infty (1+2\eta)^{-j+1} \frac{2\eps}{\Upsilon} \\
        &= \left( \log_{1+2\eta}(\nicefrac{3}\eta) + 2\frac{1+2\eta}{2\eta}\right) \frac{\eps}{\Upsilon}
        = \eps\,.
    \end{align*}
    The last equality follows from the definition of $\Upsilon = \log_{1+2\eta}(\nicefrac{3}\eta) + \frac{1+2\eta}{\eta}$.

    Finally,
    the first term $\prod_{t=1}^T \frac{\Pr[S_t=U_t\mid \rho_t=d_t, <t]}{\Pr[S_t'=U_t\mid \rho_t'=d_t, <t]}$ follows similarly to the second term
    as we only update solution using a static $(\nicefrac\eps\Upsilon, \nicefrac\delta\Upsilon)$-DP algorithm when \DensitySvt answers ``above''.
    However,
    we must handle the case where the static algorithm is satisfies approximate-DP.
    Same as before,
    write
    \begin{align*}
        \sum_{t=1}^T \ln\left( \frac{\Pr[S_t=U_t\mid \rho_t=d_t, <t]}{\Pr[S_t'=U_t\mid \rho_t'=d_t, <t]} \right)
        &=\sum_i \sum_{t=t_{i-1}+1}^{t_i} \ln\left( \frac{\Pr[S_t=U_t\mid \rho_t=d_t, <t]}{\Pr[S_t'=U_t\mid \rho_t'=d_t, <t]} \right)\,.
    \end{align*}
    Let us focus on the $i$-th inner summation,
    as we only call the static algorithm once over those timestamps.
    We call the static algorithm with privacy parameters $(\nicefrac\eps\Upsilon, \nicefrac\delta\Upsilon)$,
    But similar to the above,
    by privacy amplification via subsampling,
    the output is actually $(\frac{2q_{t_i}\eps}{\Upsilon}, \frac{q_{t_i}\delta}{\Upsilon})$-DP
    when $q_{t_i}<1$.
    Hence there is an event $\mcal E_i$ with probability $1-\frac{q_{t_i}\delta}{\Upsilon}$ over which
    \[
        \sum_{t=t_{i-1}+1}^{t_i} \ln\left( \frac{\Pr[S_t=U_t\mid \rho_t=d_t, <t]}{\Pr[S_t'=U_t\mid \rho_t'=d_t, <t]} \right)
        \leq 
        \begin{cases}
            \frac{\eps}\Upsilon, &q_{t_i}=1\,, \\
            \frac{2q_{t_i} \eps}{\Upsilon}, &q_{t_i}<1\,.
        \end{cases}
    \]
    By a union bound over all events $\mcal E_i$,
    we see that with probability
    \[
        1-\sum_{i\geq 1} \frac{q_{t_i}\delta}{\Upsilon}
        \geq 1-\left( \log_{1+2\eta}(\nicefrac{3}\eta)\cdot \frac{\delta}{\Upsilon} + \sum_{j\geq 1} (1+2\eta)^{-j+1} \frac{\delta}{\Upsilon} \right)
        \geq 1-\delta\,,
    \]
    the privacy loss variable is bounded above by
    \[
        \log_{1+2\eta}(\nicefrac{3}\eta)\cdot \frac{\eps}{\Upsilon} 
        + 2 \sum_{j\geq 1} (1+2\eta)^{-j+1} \frac{\eps}{\Upsilon}
        \leq \eps\,.
    \]

    This concludes the proof.
\end{proof}

\subsection{Approximation}\label{sec:simple-dsg-approximation}
We now proceed to demonstrate the approximation guarantees obtained by our reduction.
\begin{theorem}[Approximation Guarantee]\label{thm:simple-dsg-approximation}
    Let $\eta\in (0, \nicefrac18), c\geq 1$ be arbitrary.
    Suppose \Cref{assum:static-dp-dsg-alg} holds.
    There is an absolute constant $C > 0$ depending only on $c$
    such that with probability $1-n^{-c}$,
    for every $t\in [T]$,
    \Cref{alg:simple-dsg-main} releases a $\left( (1+56\eta)(1+\alpha), 4\kappa \right)$-approximate densest subgraph of $G_t$ for
    \[
        \kappa = \max\left( 
            C\cdot \zeta(n, \nicefrac\eps{\Upsilon}, \nicefrac\delta{\Upsilon}, \alpha), 
            \frac{2C\Upsilon \ln(n)}{\eps}
        \right)\,,\qquad\qquad
        \Upsilon = \log_{1+2\eta}(\nicefrac{3}\eta) + \frac{1+2\eta}\eta \,.
    \]
\end{theorem}
We prove \Cref{thm:simple-dsg-approximation} in three smaller steps.
In \Cref{sec:uniform-subsampling},
we verify that our subsampling scheme is valid.
Then,
\Cref{sec:conditional-approximation} shows the approximation guarantees under a simplifying condition.
This is the bulk of the technical work of this paper.
Finally,
we argue that the simplifying condition can be removed without loss of generality
in \Cref{sec:simple-dsg-approximation-proof}.

\subsubsection{Uniform Subsampling}\label{sec:uniform-subsampling}
We begin with the following adaptation of the uniform subsampling results for streaming densest subgraph.
\begin{lemma}[Implicit in \cite{EHW16}]\label{lem:dsg-uniform-subsampling}
    Let $\eta\in (0, 1)$.
    Let $G$ be a graph and let $\rho \coloneqq \max_X \rho_G(X)$ denote the density of a densest subgraph.
    Suppose $H$ is the subgraph obtained by subsampling edges with probability $q\in [0, 1]$.
    Then for every $\varnothing\neq X\sset V(G)$,
    \[
        \Pr\left[ \abs*{\frac1q \rho_H(X) - \rho_G(X)}
        \geq \frac\eta2 \cdot \rho \right]
        \leq 3\exp\left( \frac{q\cdot \card{X}\eta^2 \rho}{12} \right)
        \,.
    \]
\end{lemma}
This leads to the following variant of the main result of \cite{mcgregor2015densest,EHW16}.
\begin{theorem}\label{thm:dsg-uniform-subsamping}
    Let $G$ be a graph and let $\rho \coloneqq \max_X \rho_G(X)$ denote the density of a densest subgraph.
    Suppose $H$ is the subgraph obtained by subsampling edges with probability $q\geq \frac{60c \ln(n)}{\rho \eta^2}$
    for some $\eta\in (0, 1)$ and $c\geq 1$.
    Then with probability $1-n^{-c}$,
    every $\varnothing\neq X\sset V(G)$ satisfies
    $
        \abs*{\frac1q \rho_H(X) - \rho_G(X)}
        \geq \frac\eta2 \cdot \rho\,.
    $
\end{theorem}
Given \Cref{lem:dsg-uniform-subsampling},
the proof of \Cref{thm:dsg-uniform-subsamping} is identical to \cite{EHW16},
thus we omit the proof.
In essence,
the authors take a straightforward union bound over the $O(n^k)$ vertex sets of size $k$,
and then another union bound over the $n$ choices of $k$.
As a corollary,
we derive the following result, 
which we use for the rest of the analysis.
\begin{corollary}\label{cor:simple-dsg-target-subsample-q}
    For any $c\geq 1$,
    there is a sufficiently large constant $C\geq 1$
    such that if the subsampling probability $q_t$ of \Cref{alg:simple-dsg-main} at time $t$ satisfies
    \mbox{$q_t \geq \frac{C \ln(n)}{\rho_t \eta^2}$},
    then every $\varnothing\neq X\sset V(G)$ satisfies
    $
        \abs*{\frac1q \rho_{H_t}(X) - \rho_{G_t}(X)}
        \leq \eta\cdot \rho_t
    $
    with probabililty $1-\frac1{3Tn^c}$.
    Here $\rho_t \coloneqq \max_X \rho_{G_t}(X)$ denotes the density of a densest subgraph
    of the original graph $G_t$ at time $t$.
\end{corollary}

\subsubsection{Approximation under Simplifying Condition}\label{sec:conditional-approximation}
Let $\Gamma(n, \kappa)$ denote an arbitrary $n$-vertex $2\kappa$-regular graph.
We first prove the approximation guarantees under the simplifying condition that the initial graph is $G_0 = \Gamma(n, \kappa)$.
The general case can be handled easily once this is done.

We first identify a ``good'' event $\mcal E$ which occurs with high probability.
Conditioned upon this event,
we would like all subroutines to succeed and all noise we inject into the algorithm to be bounded.
\begin{lemma}\label{lem:simple-dsg-event-subroutines-succeed}
    For any $c\geq 1$,
    the following holds
    for a sufficiently large absolute constant $C\geq 1$ depending only on $c$:
    There is an event $\mcal E$ over the randomness of \Cref{alg:simple-dsg-main} 
    which occurs with probability $1-\frac1{3n^c}$.
    Conditioned on $\mcal E^{(1)}$,
    \begin{enumerate}[(i)]
        \item \DensitySvt answers the $t$-th query within additive error $\frac{C\Upsilon \log(n)}\eps$ for every timestep $t\in [T]$,
        \item for each timestep $t\in [T]$ where the static algorithm \PrivateDensestSubgraph (\Cref{assum:static-dp-dsg-alg}) is called,
        it returns a $\left( 1+\alpha, C\cdot \zeta\left( n, \nicefrac{\eps}{\Upsilon}, \nicefrac{\delta}{\Upsilon}, \alpha \right) \right)$-approximate densest subgraph of the sparsified graph $H_t$.
    \end{enumerate}
\end{lemma}

\begin{proof}
    By \Cref{fact:laplace-concentration} and a union bound over at most $T = \poly(n)$ timesteps,
    we can choose $C$ sufficiently large so that
    the noise drawn in \DensitySvt (\Cref{alg:sparse vector technique}) is bounded by the desired quantity
    with probability $1-\frac1{6n^c}$.
    Similarly,
    by \Cref{assum:static-dp-dsg-alg} and a union bound over at most $T$ timesteps,
    we can choose $C$ sufficiently large so that
    every call to \PrivateDensestSubgraph produces the desired approximation
    with probability $1-\frac1{6n^c}$.

    Let $\mcal E$ be the intersection of the events above.
    By a union bound,
    $\mcal E$ occurs with probability $1-\frac1{3n^c}$,
    as desired.
\end{proof}

We now prove the approximation guarantees conditioned on the good event $\mcal E$ from \Cref{lem:simple-dsg-event-subroutines-succeed}.
\begin{lemma}\label{lem:simple-dsg-approximation}
    Let $\eta\in (0, \nicefrac18)$,
    $c\geq 1$ be arbitrary 
    and $C$ be a sufficiently large constant to satisfy \Cref{lem:simple-dsg-event-subroutines-succeed,cor:simple-dsg-target-subsample-q}.
    Moreover,
    assume the initial graph $G_0 = \Gamma(n, \kappa)$.
    Then,
    there are events $\mcal E^{(t)}$ such that conditioned on $\mcal E, \mcal E^{(0)}, \dots, \mcal E^{(t-1)}$,
    the event $\mcal E^{(t)}$ occurs with probability $1-\frac1{3Tn^c}$.
    Moreover,
    the following holds over $\mcal E^{(t)}$:
    \begin{enumerate}[(a)]
        \item If $q_t < 1$, subsampling probability $q_t$ satisfies
            $
                \frac{C \ln(n)}{\rho_t \eta^2}
                \leq \frac{\kappa}{\rho_t \eta}
                \leq q_t
                \leq \frac{32\kappa}{\rho_t \eta}\,,
            $
        \item $\rho_\DP^{(t)}$ is a $\left( 1+6\eta, \frac{C\Upsilon\ln(n)}\eps \right)$-approximation for $\rho_t$.
        \item $S_\DP^{(t)}$ induces a $\left( (1+56\eta)(1+\alpha), \frac{3C\Upsilon\ln(n)}\eps \right)$-approximate densest subgraph of $G_t$.
    \end{enumerate}
\end{lemma}
We note that \Cref{lem:simple-dsg-approximation} implies the desired result once we take a union bound over all events $\mcal E, \mcal E^{(1)}, \dots, \mcal E^{(T)}$.
It is stated in a way that is suitable for proof via induction
and for re-use in the upcoming space complexity analysis.
The proof of \Cref{lem:simple-dsg-approximation} we present below is conceptually simple,
as we need only ensure that our intuition of lazy updates carries through in the presence of noise.
\begin{proof}
    We argue by induction on $t$.
    The claims hold at $t=0$ since we set $q_0 = 1$,
    $\rho_\DP^{(0)} = \kappa$, which is the exactly the density of the densest subgraph at time $t=0$,
    and $S_\DP^{(0)} = V$, which is a densest subgraph of $G_0$.
    Thus we can take $\mcal E^{(0)}$ to be the trivial event of measure 1.

    Suppose inductively that the claim holds up to time $t-1$.    
    In the following,
    We prove each claim separately.

    \myparagraph{Subsampling Probability}
    Note that $q_t = \min\left( 1, \frac{3\kappa}{\rho_\DP^{(t-1)} \eta} \right)$ since $q_t$ is updated at the end of the previous edge update.
    By the induction hypothesis,
    $\rho_\DP^{(t-1)}$ is a $\left( 1+6\eta, \frac{C\Upsilon\ln(n)}\eps \right)$-approximation of $\rho_{t-1}$.
    In particular,
    \begin{align*}
        \rho_\DP^{(t-1)} 
        &\leq (1+6\eta) \rho_{t-1} + \frac{C\Upsilon\ln(n)}\eps \leq 3\rho_{t-1} \\
        \rho_\DP^{(t-1)} 
        &\geq \frac1{1+6\eta} \rho_{t-1} - \frac{C\Upsilon\ln(n)}\eps \\
        &\geq \frac1{1+6\eta} \left( \rho_{t-1} - \frac{7C\Upsilon\ln(n)}{4\eps} \right) \\
        &\geq \frac1{8} \rho_{t-1}\,. &&\text{since $\rho_{t-1}\geq \kappa\geq \frac{2C\Upsilon \ln(n)}{\eps}, \eta\in (0,\nicefrac18)$}
    \end{align*}
    Thus if $q_t < 1$,
    we have
    \begin{align*}
        q_t
        &\geq \frac{3\kappa}{3\rho_{t-1}\eta} 
        \geq \frac{C\ln(n)}{\rho_{t}\eta^2}\,, &&\text{since $\kappa\geq \frac{C\Upsilon \ln(n)}\eps \geq \frac{C \ln(n)}{\eps\eta}$} \\
        q_t 
        &\leq \frac{3 \kappa}{\rho_{t-1} \eta}
        \leq \frac{24 \kappa}{(\rho_t-1)\eta}
        \leq \frac{24 \kappa}{(1-\eta)\rho_t \eta}\,. &&\text{since $\rho_t\geq \rho_0 = \kappa > \frac1\eta$}
    \end{align*}
    By the choice of $\eta\in (0, \nicefrac18)$,
    we have the desired result.

    \myparagraph{Density Estimate}
    We see that the subsampling probability $q_t$ satisfies the assumptions of \Cref{cor:simple-dsg-target-subsample-q}.
    Hence conditioned on $\mcal E, \mcal E^{(0)}, \dots, \mcal E^{(t-1)}$,
    let $\mcal E^{(t)}$ be the event such that
    the exact maximum density $r_t$ of $H_t$ that we compute satisfies
    \begin{align}
        \abs*{\frac1{q_t} r_t - \rho_t}\leq \eta \rho_t\,. \label[ineq]{ineq:subsampling-scaled-error}
    \end{align}
    Note that the event $\mcal E^{(t)}$ occurs with probability $1-\frac1{3Tn^c}$
    Our goal is to show that $\rho_\DP^{(t)}$ remains a good approximation of the maximum density $\rho_t$.
    This is achieved by separately analyzing the cases 
    whether \DensitySvt returned ``above'' or ``below''
    and checking the desired approximation holds in both cases.

    Before diving into the case work,
    we derive some useful estimates that we will repeated use.
    By the definition of $q_t$,
    the induction hypothesis,
    and the choice of $\eta\in (0, \nicefrac18)$,
    we have
    \begin{align}
        \frac{C\Upsilon \ln(n)}{q_t \eps}
        &= \frac{\rho_\DP^{(t-1)} \eta}{3\kappa}\cdot \frac{C\Upsilon \ln(n)}{\eps}
        \leq \frac\eta3 \rho_\DP^{(t-1)} \label[ineq]{ineq:additive-error-by-density-estimate} \\
        &\leq \frac\eta3 \left[ (1+6\eta) \rho_{t-1} + \frac{C\Upsilon \ln(n)}{\eps} \right]
        \leq \eta \rho_{t-1} + \frac{C\Upsilon \ln(n)}{\eps}\,. \label[ineq]{ineq:additive-error-by-density}
    \end{align}
    
    Now suppose that \DensitySvt returns ``below''
    so that we keep $\rho_\DP^{(t)} = \rho_\DP^{(t-1)}$.
    We wish to upperbound $\rho_t$ in terms of $\rho_\DP^{(t)}$.
    Indeed,
    by the definition of \DensitySvt and $\mcal E$ (\Cref{lem:simple-dsg-event-subroutines-succeed}),
    \begin{align}
        r_t \leq q_t(1+2\eta) \rho_{\DP}^{(t)} + \frac{C\Upsilon\ln(n)}{\eps}\,. \label[ineq]{ineq:density-svt-below-error}
    \end{align}
    Thus by \Cref{ineq:subsampling-scaled-error,ineq:density-svt-below-error,ineq:additive-error-by-density-estimate},
    \begin{align*}
        (1-\eta) \rho_t
        &\leq \frac{r_t}{q_t}
        \leq (1+2\eta) \rho_{\DP}^{(t)} + \frac{C\Upsilon\ln(n)}{q_t\eps}
        \leq (1+3\eta) \rho_\DP^{(t)}\,.
    \end{align*}
    We also need to lower bound $\rho_t$ in terms of $\rho_\DP^{(t)}$.
    To see this,
    by the induction hypothesis,
    we have
    \begin{align*}
        \rho_{\DP}^{(t)} 
        = \rho_{\DP}^{(t-1)}
        \leq (1+6\eta) \rho_{t-1} + \frac{C\Upsilon\ln(n)}{\eps}
        \leq (1+6\eta) \rho_t + \frac{C\Upsilon\ln(n)}{\eps}\,. &&\text{since $\rho_{t-1}\leq \rho_t$}
    \end{align*}

    Now suppose $\DensitySvt$ returns ``above''
    so that the algorithm sets $\rho_\DP^{(t)} = \max\set*{(1+2\eta) \rho_\DP^{(t-1)}, \frac{\tilde r_t}{q_t}}$.
    Following similar steps,
    we upper bound $\rho_t$ in terms of $\rho_\DP^{(t)}$.
    We use the definition of $\mcal E$ to deduce that
    \begin{align}
        \abs{\tilde r_t - r_t}
        &\leq \frac{C\Upsilon\ln(n)}{\eps}\,. \label[ineq]{ineq:sparsified-density-estimate}
    \end{align}
    Then,
    by \Cref{ineq:sparsified-density-estimate,ineq:additive-error-by-density-estimate},
    we see that
    \begin{align*}
        \rho_\DP^{(t)}
        &\geq \frac{\tilde r_t}{q_t}
        \geq \frac{r_t}{q_t} - \frac{C\Upsilon \ln(n)}{q_t \eps}
        \geq (1-2\eta) \rho_t - \frac{C\Upsilon\ln(n)}\eps\,.
    \end{align*}
    We also need to provide a lower bound of $\rho_t$ in terms of $\rho_\DP^{(t)}$,
    which is slightly more complicated due taking $\rho_\DP^{(t)}$ to be a maximum over two terms.
    By the definition of \DensitySvt and $\mcal E$,
    we have that
    \begin{align}
        r_t &\geq q_t(1+2\eta) \rho_{\DP}^{(t-1)} - \frac{C\Upsilon\ln(n)}{\eps}\,. \label[ineq]{ineq:svt-above-error}
    \end{align}
    Thus by \Cref{ineq:subsampling-scaled-error,ineq:svt-above-error,ineq:additive-error-by-density-estimate},
    \begin{align*}
        (1+\eta)\rho_t
        &\geq \frac{r_t}{q_t}
        \geq (1+2\eta) \rho_\DP^{(t-1)} - \frac{C\Upsilon \ln(n)}{q_t \eps}
        \geq (1+\eta) \rho_\DP^{(t-1)}\,.
    \end{align*}
    In addition,
    by \Cref{ineq:sparsified-density-estimate,ineq:additive-error-by-density-estimate},
    \begin{align*}
        \frac{\tilde r_t}{q_t}
        &\leq \frac{r_t}{q_t} + \frac{C\Upsilon \ln(n)}{q_t \eps}
        \leq (1+\eta) \rho_t + \frac{C\Upsilon\ln(n)}\eps \,.
    \end{align*}
    Putting the two upper bounds together yields
    \[
        \rho_\DP^{(t)}
        = \max\set*{(1+2\eta) \rho_\DP^{(t-1)}, \frac{\tilde r_t}{q_t}}
        \leq (1+3\eta) \rho_t + \frac{C\Upsilon\ln(n)}\eps\,.
    \]

    All in all,
    for $\eta\in (0, \nicefrac18)$,
    using the inquality $1-x\geq \frac1{1+2x}$ for $x\in [0, \nicefrac12]$,
    we ensure that
    \[
        \frac1{1+6\eta} \rho_t - \frac{C\Upsilon \ln(n)}\eps
        \leq \rho_\DP^{(t)} 
        \leq (1+6\eta) \rho_t \frac{C\Upsilon \ln(n)}\eps\,.
    \]

    \myparagraph{Approximate Solution}
    Fix a timestamp $t$ where \PrivateDensestSubgraph is called.
    By \Cref{lem:simple-dsg-event-subroutines-succeed},
    we know that $\rho_{H_t}(S_\DP^{(t)}) \geq \frac1{1+\alpha} \OPT_{H_t} - C\cdot \zeta(n, \nicefrac{\eps}{\Upsilon}, \nicefrac\delta{\Upsilon}, \alpha)$
    for any $\varnothing\neq X\sset V$.
    Conditioning on $\mcal E, \mcal E^{(0)}, \dots, \mcal E^{(t)}$,
    either $q_t = 1$ and there is nothing to prove,
    or $q_t < 1$ and
    \begin{align*}
        \rho_{G_t}(S_\DP^{(t)}) + \eta \rho_t
        &\geq \frac1{q_t} \rho_{H_t}(S_\DP^{(t)}) \\
        &\geq \frac1{1+\alpha}\cdot \frac1{q_t} \OPT_{H_t} - \frac1{q_t}\cdot C\cdot \zeta\left( n, \frac{\eps}{\Upsilon}, \frac{\delta}{\Upsilon}, \alpha \right) \\
        &\geq \frac{1-\eta}{1+\alpha} \rho_t - \frac{\rho_t \eta}{\kappa}\cdot C\cdot \zeta\left( n, \frac{\eps}{\Upsilon}, \frac{\delta}{\Upsilon}, \alpha \right) \\
        &\geq \frac{1-\eta}{1+\alpha} \rho_t - \eta \rho_t &&\text{$\kappa\geq C\cdot \zeta(n, \nicefrac\eps{\Upsilon}, \nicefrac\delta{\Upsilon}, \alpha)$} \\
        \rho_{G_t}(S_\DP^{(t)})
        &\geq \frac{1-4\eta}{1+\alpha} \rho_t\,. &&\text{$\alpha\in [0, 1]$}
    \end{align*}

    Finally,
    we note that we output a new solution whenever $\rho_\DP$ is updated.
    Since $\rho_\DP$ is a $\left( 1+6\eta, \frac{C\Upsilon \ln(n)}\eps \right)$-approximation of the density,
    $\rho_t$ can increase by at most $(1+6\eta)^2 + (1+6\eta)\frac{2C\Upsilon\ln(n)}\eps$ before there must be an update.
    Hence for all timestamps,
    \begin{align*}
        \rho_{G_t}(S_\DP^{(t)})
        &\geq \frac{1-4\eta}{(1+6\eta)^2(1+\alpha)} \rho_t - \frac{3C\Upsilon\ln(n)}\eps \\
        &\geq \frac1{(1+8\eta)^3(1+\alpha)} \rho_t  - \frac{3C\Upsilon\ln(n)}\eps &&\text{$\eta\in (0, \nicefrac18)$}\\
        &\geq \frac1{(1+56\eta)(1+\alpha)} \rho_t - \frac{3C\Upsilon\ln(n)}\eps\,.
    \end{align*}
\end{proof}

\subsubsection{Proof of \texorpdfstring{\Cref{thm:simple-dsg-approximation}}{Theorem}}\label{sec:simple-dsg-approximation-proof}
We are now ready to prove \Cref{thm:simple-dsg-approximation}.
\begin{proof}[Proof of \Cref{thm:simple-dsg-approximation}]
    If the initial graph $G_0 $ is precisely some $n$-vertex $\kappa$-regular graph $\Gamma(n, \kappa)$,
    then the result follows from \Cref{lem:simple-dsg-approximation}
    as $\kappa\geq \frac{C\Upsilon\ln(n)}\eps$.
    Otherwise,
    we can think of the algorithm as operating on the modified graph $G_t'=G_t\cup \Gamma(n, \kappa)$,
    where $G_t'$ is obtained by taking the union of edges of $G_t$ and $\Gamma(n, \kappa)$.
    Note that for any $\varnothing\neq X\sset V$,
    $\abs{\rho_{G_t}(X) - \rho_{G_t'}}\leq \kappa$
    since each vertex differs in at most $\kappa$ edges between $G_t, G_t'$.
    Hence we incur at most an additional additive error of $\kappa$ for general graphs,
    as required.
\end{proof}

\subsection{Space}\label{sec:simple-dsg-space}
Finally,
we formally demonstrate the space guarantees.
\begin{theorem}[Space Usage]\label{thm:simple-dsg-space}
    Let $\eta\in (0, \nicefrac18)$,
    $c\geq 1$ be arbitrary 
    and $C$ be a sufficiently large constant to satisfy \Cref{lem:simple-dsg-event-subroutines-succeed,cor:simple-dsg-target-subsample-q}.
    There is a sufficiently large absolute constant $C'\geq 1$ depending only on $c$
    such that probability $1-n^{-c}$,
    \Cref{alg:simple-dsg-main} uses space at most
    \[
        \max\left( 
            \frac{32C'Cn}\eta\cdot \zeta(n, \nicefrac\eps{\Upsilon}, \nicefrac\delta{\Upsilon}, \alpha), 
            \frac{64C'C \Upsilon n\ln(n)}{\eta^2 \eps}
        \right)\,, \qquad\qquad
        \Upsilon = \log_{1+2\eta}(\nicefrac{3}\eta) + \frac{1+2\eta}\eta \,.
    \]
\end{theorem}

\begin{remark}\label{rem:counting-space}
    Similar to \cite{epasto2024power},
    we store a random number $h_e\sim U[0, 1]$ 
    uniformly drawn between $[0, 1]$
    for each edge $e$ in order to simplify the presentation.
    In order to fit this into 1 word,
    we can approximately implement $h$ as drawing a uniform random integer $\nu\in \set{1, 2, \dots, \poly(n)}$
    and taking $h = \nicefrac1\nu$.
    Alternatively,
    we can implement the equivalent algorithm that draws a biased coin to decide whether to keep an edge
    after some update.
    Moreover,
    each edge fits into 1 word,
    hence we need only bound the total number of stored edges.
\end{remark}

In light of \Cref{rem:counting-space},
we bound the number of edges we store with high probability.
\begin{proof}[Proof of \Cref{thm:simple-dsg-space}]
    Let $\Gamma(n, \kappa)$ denote an arbitrary $n$-vertex $2\kappa$-regular graph.
    It suffices to prove the statement for the overlayed graph $G_t' = G_t\cup \Gamma(n, \kappa)$,
    as the space usage can only increase with more edges.
    Conditional on the events $\mcal E, \mcal E^{(0)}, \dots, \mcal E^{(T)}$ from \Cref{lem:simple-dsg-event-subroutines-succeed,lem:simple-dsg-approximation},
    we have
    \[
        q_t 
        \leq \frac{32\kappa}{\rho_t \eta}
        = \frac{32}{\rho_t \eta} \max\left( C\cdot\zeta(n, \nicefrac\eps\Upsilon, \nicefrac\delta\Upsilon, \alpha), \frac{2C\Upsilon \ln(n)}\eps \right)
    \]
    for each $t$.
    By a multiplicative Chernoff bound (\Cref{thm:multiplicative-chernoff}),
    the number of subsampled edges at time $t$ is at most $\frac{32C'\kappa m_t}{\rho_t}$ with probability $1-\frac1{3Tn^c}$
    for some absolute constant $C'$ depending only on $c$.
    Note that $\rho_t\geq \frac{m_t}{n}$ so that the claim holds for time $t$.
    By a union bound over all $T$ timestamps and the events $\mcal E, \mcal E^{(0)}, \dots, \mcal E^{(T)}$,
    the claim holds for all $t$ with probability $1-n^{-c}$.
\end{proof}

\section{Conclusion}
We provide the first formal connection between two applications of subsampling: sparsification and privacy amplification,
allowing us to close the gap between the utility of continual release densest subgraph algorithms
compared to their static DP counterparts
as well as the gap between their space complexity
compared to the non-private streaming counterparts.
It would be very interesting to extend the connection between the two goals of sublinear space and privacy amplification to other problems.
Our algorithmic improvements can all be interpreted as performing some form of graph densification,
a concept that is so far unexplored in differential privacy.
We believe this idea may be applicable to other graph privacy problems.

It would also be interesting to explore space complexity lower bounds in the continual release setting,
as our space usage depends on $\frac1\eps$,
and it is not clear if this requirement is intrinsic.
Finally,
there remains a $\Omega(\sqrt{\frac1\eps\log n})$ gap between the best upper and lower bounds in the additive error for DP densest subgraph,
even for privately estimating the maximum density value in the static setting.

%% file: appendix.tex
\section{Deferred Preliminaries}\label{apx:prelims}
We now remind the reader of some standard definitions and tools.

\subsection{Approximation Algorithms}
We consider optimization problems in both the minimization and maximization setting. Let $\OPT \in \R$ be the optimum value for the problem. For a minimization problem,
a \emph{$(\alpha, \zeta)$-approximate algorithm} outputs a solution with cost at most $\alpha\cdot \OPT + \zeta$.
For a maximization problem,
$(\alpha, \zeta)$-approximate algorithm outputs a solution of value at least $\frac1\alpha\cdot \OPT - \zeta$.
For estimation problems, a $(\alpha, \zeta)$-approximate algorithm outputs an estimate $\tilde R\in \R$ of some quantity $R\in \R$
where $\frac1\alpha R - \zeta \leq \tilde R \leq \alpha \cdot R + \zeta$. 

As a shorthand,
we write $\alpha$-approximation algorithm to indicate a $(\alpha, 0)$-approximation algorithm. Our approximation bounds hold, with high probability, 
{over all} releases.

\subsection{Concentration Inequalities}
Below, we state the concentration inequalities we use throughout this work.

\begin{theorem}[Multiplicative Chernoff Bound; Theorems 4.4, 4.5 in \cite{mitzenmacher2015probability}]\label{thm:multiplicative-chernoff}
    Let $X = \sum_{i = 1}^n X_i$ where each $X_i$ is a Bernoulli variable which takes value $1$ with probability $p_i$ and value $0$
    with probability $1-p_i$. Let $\mu = \expect[X] = \sum_{i = 1}^n p_i$.
    The following holds.
    \begin{enumerate}
        \item Upper Tail: $\prob[X \geq (1+\psi) \cdot \mu] \leq \exp\left(-\frac{\psi^2\mu}{2 + \psi}\right)$ for all $\psi > 0$;
        \item Lower Tail: $\prob[X \leq (1-\psi) \cdot \mu] \leq \exp\left(-\frac{\psi^2\mu}{3}\right)$ for all $0 < \psi < 1$.
    \end{enumerate}
\end{theorem}

\begin{fact}[Fact 3.7 in \cite{dwork2014algorithmic}]\label{fact:laplace-concentration}
    Suppose $Y\sim \Lap(b)$.
    Then
    $
        \Pr[\abs{Y}\geq t\cdot b]
        \leq \exp(-t)\,.
    $
\end{fact}

\section{Private Static Densest Subgraph}\label{apx:static-dp-dsg-alg}
\begin{theorem}[Theorem 6.1 in \cite{DLL23}]\label{thm:improved private static densest subgraph}
    There is an $\eps$-edge DP densest subgraph algorithm
    that uses $\poly(n)$ time and $O(m+n)$ space
    to return a subset of vertices
    that induces a $\left( 2, O\left( \frac1\eps \log n \right) \right)$-approximate densest subgraph
    with probability at least $1-\poly(\nicefrac1n)$.
\end{theorem}

\begin{theorem}[Theorem 5.1 in \cite{DLRSSY22}]\label{thm:private static densest subgraph}
    Fix $\eta \in (0, 1)$.
    There is an $\eps$-edge DP densest subgraph algorithm
    that uses $\poly(n)$ time and $O(m+n)$ space
    to return a subset of vertices
    that induces a $\left( 1+\eta, O\left( \frac{\log^4 n}{\eps \eta^3} \right) \right)$-approximate densest subgraph
    with probability at least $1-\poly(\nicefrac1n)$.
\end{theorem}

\begin{theorem}[Theorem 3.6 in \cite{dinitz2024tight}]\label{thm:static-approx-dp-dsg-alg}
    There is an $(\eps, \delta)$-edge DP densest subgraph algorithm
    that uses $\poly(n)$ time and $O(m+n)$ space
    to return a subset of vertices
    that induces a $\left( 1, O\left( \frac1\eps \sqrt{\log(n) \log(\nicefrac{n}\delta)} \right) \right)$-approximate densest subgraph
    with probability at least $1-\poly(\nicefrac1n)$.
\end{theorem}

\section{The Density of the Densest Subgraph}\label{apx:densest value}
We now recall the following theorem about computing the exact density of a densest subgraph in the static setting
as stated in \cite{epasto2024power}.
\begin{theorem}\label{thm:densest value}
    Given an undirected graph $G=(V, E)$,
    there is an algorithm \textsc{ExactDensity} that computes the exact density of the densest subgraph.
    Moreover,
    the algorithm terminates in $\poly(n)$ time
    and uses $O(n+m)$ space.
\end{theorem}